\documentclass[letterpaper,11pt]{article}
\usepackage{verbatim}
\usepackage{xspace}
\usepackage{amsmath, amsthm ,amsfonts,graphicx}
\usepackage{algorithm}
\usepackage{algorithmicx}
\usepackage{algpseudocode}
\usepackage{multirow}
\usepackage{caption}
\usepackage{rotating}

\newtheorem{thm}{Theorem}[section]

\newtheorem{lem}[thm]{Lemma}

\newtheorem{defn}[thm]{Definition}
\newtheorem{rem}[thm]{Remark}

\newcommand{\IC}{\ensuremath\mathrm{IC}}
\newcommand{\IR}{\ensuremath\mathrm{IR}}
\newcommand{\NPT}{\ensuremath\mathrm{NPT}}

\title{Optimal Deterministic Auctions with Correlated Priors}

\author{Christos Papadimitriou\thanks{Research (partially) supported by NSF grant CCF-0964033 and by a Google University Research Award.}\\{\small UC Berkeley}\\{\small\texttt{christos@cs.berkeley.edu}} \and\and\and George Pierrakos\footnotemark[\value{footnote}]\\{\small UC Berkeley}\\{\small\texttt{georgios@cs.berkeley.edu}}}
\date{}
\begin{document}

\maketitle
\thispagestyle{empty}

\abstract{We revisit the problem of designing the profit-maximizing single-item auction, solved by Myerson in his seminal paper for the case in which bidder valuations are independently distributed.  We focus on general joint distributions, seeking the optimal deterministic incentive compatible auction.  We give a geometric characterization of the optimal auction, resulting in a duality theorem and an efficient algorithm for finding the optimal deterministic auction in the two-bidder case and an NP-completeness result for three or more bidders.}

\section{Introduction}

Myerson's paper~\cite{RePEc:nwu:cmsems:362} on optimal auction design is remarkable in several ways.  It is not the first important paper on auctions of course~\cite{vickrey61a}, but it pioneers the point of view of its title: {\em auction design,} that is, the exploration and evaluation of a large design space in a mindset that is very much one of computer science.   It is elegant and methodologically original and powerful, completely resolving this important and difficult problem and providing insights, lemmata, and methodology that would be useful in related contexts for decades. Myerson considers the auction of an item to bidders with private valuations whose prior distributions are independent and known, and seeks incentive compatible  mechanisms that maximize auctioneer revenue.  His simple and elegant solution involves the {\em virtual valuation}, a quantity each bidder can compute from their private valuation and the valuation distribution through a mathematical maneuver called {\em ironing}.  

Myerson leaves open the case in which the valuations are correlated; in subsequent work, Cr\'{e}mer and McLean~\cite{RePEc:ecm:emetrp:v:56:y:1988:i:6:p:1247-57,RePEc:ecm:emetrp:v:53:y:1985:i:2:p:345-61} consider correlated valuations and solve the problem for the case where auctions are only required to be {\em interim} individually rational. In fact, in this framework the uncorrelated case is a singularity in the sense that in most cases (when the correlation has ``full rank'' in a certain precise sense) full surplus can be extracted in expectation through appropriate offers of lotteries to the bidders.  Despite the elegance of their result, the fact that bidders may be charged merely for participating in the auction--lottery (including losers) has been criticized as rendering the auction impractical~\cite{RePEc:ecm:emetrp:v:60:y:1992:i:2:p:395-421}, especially for settings where agents may easily cancel their participation after the auction is conducted. It is therefore of tantamount importance to consider the question of designing the optimal {\em ex post} individually rational auction for correlated valuations.  This is the question we consider, and in a sense completely answer, in the current paper.

This past decade saw the advent and rise of Algorithmic Game Theory~\cite{DBLP:conf/icalp/Papadimitriou01,DBLP:conf/stoc/NisanR99}, a research tradition which can be seen as a complexity-theoretic critique of Mathematical Economics, with Internet in the backdrop.  This point of view has yielded a host of important results and new insights, for example related to the complexity of equilibria~\cite{DBLP:conf/focs/ChenD06,DBLP:conf/stoc/DaskalakisGP06}, the trade-offs between complexity, approximation, and incentive-compatibility in social welfare-maximizing mechanism design~\cite{DBLP:conf/stoc/NisanR99,DBLP:conf/focs/PapadimitriouSS08}, and (in an extended sense that includes on-line algorithms as a part of complexity theory) the price of anarchy~\cite{DBLP:conf/stacs/KoutsoupiasP99,DBLP:journals/jacm/RoughgardenT02,DBLP:journals/jcss/Roughgarden03}.  However, there has been little progress in looking at Bayesian auctions \`a-la Myerson from this point of view.  Ronen~\cite{DBLP:conf/sigecom/Ronen01} came up with a mechanism for the correlated case that achieves half of the optimum revenue, while Ronen and Saberi~\cite{DBLP:conf/focs/RonenS02} showed that no ``ascending auction'' can do better than 7/8, and they conjectured that all relevant auctions are ascending (incidentally, we disprove this conjecture by showing that the optimal two-bidder auction may not be ascending).  Missing from these two papers, however, is a concrete sense of the ways in which this is a difficult problem. We provide this here.

There were several follow-up papers ~\cite{DBLP:conf/sigecom/ChawlaHK07,DBLP:conf/sigecom/HartlineR09,DBLP:conf/sigecom/DhangwatnotaiRY10}, examining the extent to which simple mechanisms can achieve good approximations of Myerson's optimal mechanism, and motivated by Myerson's astonishing result that the optimal mechanism for the regular i.i.d. setting is simply a second-price auction with reserve prices.  For the most part, these are positive results yielding constant approximations by simple mechanisms in a variety of settings.  Another line of work is prior-free mechanism design, where the goal is to design mechanisms that achieve profits comparable to that of some well-behaved benchmark~\cite{Goldberg02competitiveauctions}. This direction became especially interesting after Hartline and Roughgarden developed a framework in~\cite{DBLP:conf/stoc/HartlineR08} that is grounded in Bayesian optimal mechanism design allowing one to design mechanisms that simultaneously approximate all Bayesian optimal mechanisms. The intermediate approach of having bidders' valuations coming from a distribution that is nonetheless {\em unknown} to the auctioneer has also been considered~\cite{DBLP:conf/sigecom/DhangwatnotaiRY10}.  Finally the important open problem of designing optimal multi-parameter mechanisms has been addressed for some distributions of preferences, such as additive valuations with budget constraints~\cite{DBLP:conf/stoc/BhattacharyaGGM10} and unit-demand settings~\cite{DBLP:conf/stoc/ChawlaHMS10} and connections of this problem to the algorithmic problem of optimal pricing have been made~\cite{DBLP:conf/sigecom/ChawlaHK07,DBLP:conf/soda/GuruswamiHKKKM05}; economists have also made some attempts to extend Myerson's results to multi-item auctions \cite{RePEc:bla:restud:v:67:y:2000:i:3:p:455-81,RePEc:eee:jetheo:v:134:y:2007:i:1:p:494-512}.

In this paper we take a complexity-theoretic look at the general, correlated valuations case of Myerson's single-item auctions.  We point out that the optimal auction design problem can be reduced essentially to a maximum weighted independent set problem in a particular graph whose vertices are all possible tuples of valuations (an uncountable set, of course, in the continuous case).  If the distribution is discrete, this is an ordinary graph-theoretic problem; no such combinatorial characterization had been known, and this had been the main difficulty in developing an algorithmic and complexity-theoretic understanding of the problem.  For discrete distributions, this leads directly to an efficient algorithm in the case of two bidders, where the graph is bipartite, while in the case of three or more bidders NP-completeness (in fact, inapproximability) prevails.  For continuous distributions, we prove a duality characterization through a Monge-Kantorovich-like problem~\cite{Evans99partialdifferential}, and from this a fully polynomial  approximation scheme for two bidders when the distribution is continuous enough and accessible through an oracle --- plus, as an aside, a 2/3 approximation for three bidders, improving the previously best known approximation of \cite{DBLP:conf/sigecom/Ronen01}.

Our results rest on a geometric characterization of optimal deterministic auctions. An important element of our proof is the so-called {\em marginal profit contribution} function; it bears some similarities to Myerson's virtual valuation function \cite{RePEc:nwu:cmsems:362}, the most important of them being that they both admit a marginal revenue interpretation in the spirit of \cite{RePEc:ucp:jpolec:v:97:y:1989:i:5:p:1060-90}. However, despite their similarities and their somewhat common derivation, marginal profit contribution functions are different from Myerson's virtual valuations in a number of ways: they only take positive values, they are not necessarily monotone and they do not admit a natural interpretation as valuations in some modified domain.  One important ingredient of Myerson's approach to the design of optimal auctions is an analytical maneuver he calls {\em ironing;} Myerson uses ironing to transform a potentially non-monotone allocation rule into a monotone one, without hurting revenue. Our approach circumvents ironing by restricting the space of mechanisms explored; we achieve that by imposing an additional technical condition which limits the design space into a subset of all mechanisms, but one which still contains all the optimal mechanisms.

The work most related to ours is that of Dobzinski, Fu and Kleinberg~\cite{DBLP:journals/corr/abs-1011-2413}, who also study the problem of designing the optimal auction for the correlated setting.  They obtain a collection of very interesting results, which however are quite complementary to ours:  Based, among others, on insights from~\cite{DBLP:conf/sigecom/Ronen01}, they arrive at efficient algorithms for computing the optimal randomized mechanism that is truthful in expectation, and a constant factor approximation of the optimal deterministic auction for any number of bidders.  

\section{Preliminaries}\label{sec:characterization}

Imagine $n$ bidders seeking an indivisible good offered in auction. We assume that each bidder has a private valuation $v_i$ for the item and that bidders' valuations are drawn from some joint distribution over $[0,1]^n$ whose density function we denote by $\phi(\bf{v})$.  We consider both discrete and continuous $\phi$.  Discrete distributions are the source of the combinatorial insights underlying our approach, while continuous distributions provide continuity with the spirit and methodology of Myerson's paper, another important source of inspiration.   In the continuous case, we follow Myerson in making the analytically convenient assumption that $\phi(x)>0$ for all $x\in[0,1]^n$.  This is hardly a loss of generality, since a small minimum value on every point can be achieved by changing $\phi$ very little.  In stating an algorithm for the two-dimensional continuous case (Section~\ref{sec:continuous}), we shall also assume that $\phi$ is Lipschitz-continuous and accessible through an oracle (in such a way that, for example, it can be approximately integrated over nice regions). 

In the discrete case, let $Sup(\phi)$ denote the finite support of the joint discrete distribution $\phi$. Then $\phi$ is presented as a finite set of $|Sup(\phi)|$ $(n+1)$-tuples of the form $(v_1^i,v_2^i,\ldots,v_n^i,\phi^i)$, one for each point $v^i\in Sup(\phi)$, where $\phi^i$ is the probability mass concentrated at the point $(v^i)$ of the support. 

We are interested in designing auctions that maximize the auctioneer's profit. Formally, an auction consists of an allocation rule denoted by $x_i(\bf{v})$, the probability of bidder $i$ getting allocated the item, and a payment rule denoted by $p_i(\bf{v})$ which is the price paid by bidder $i$.  In this paper we focus our attention on deterministic mechanisms so that $x_i({\bf v})\in \{0,1\}$. Our goal is to maximize the auctioneer's expected profit $\sum_{i=1}^n\mathbb{E}\left[\phi({\bf v})p_i({\bf v})\right]$.

We want our auction to be ex-post incentive compatible (IC) and individually rational (IR), with no positive transfers (NPT). These notions are defined as follows:

\begin{itemize}
\item[IC:] $\forall i,v_i,v'_i,v_{-i}: v_ix_i(v_i,v_{-i})-p_i(v_i,v_{-i})\geq v_ix_i(v'_i,v_{-i})-p_i(v'_i,v_{-i})$
\item[IR:] $\forall i,v_i,v_{-i}: v_ix_i(v_i,v_{-i}) - p_i(v_i,v_{-i})\geq 0$
\item[NPT:]$\forall i,v_i,v_{-i}:  p_i(v_i,v_{-i})\geq 0$ 
\end{itemize}

We say that an allocation function $x_i(v_i,v_{-i})$ is {\em monotone} if $v_i\geq v'_i$ implies $ x_i(v_i,v_{-i})\geq x_i(v'_i,v_{-i})$ for all $i,v_{-i}$. For the case of deterministic mechanisms that we are interested in monotonicity implies that $x_i(v_i,v_{-i})$ is a step-function. The {\em threshold value} of such a step-function allocation rule is set to be the minimum winning valuation for every player given the valuations of the other bidders: $t_i(v_{-i}) = \inf\{v_i\in [0,1]\mid x_i(v_i,v_{-i})=1\}$.

The following theorem provides a characterization of mechanisms for the setting we are interested in; a proof can be found in \cite{Nisan2007}.

\begin{thm}\label{thm:characterization}
A deterministic mechanism satisfies $\IC$, $\IR$ and $\NPT$  if and only if the following conditions hold.
\begin{enumerate}
\item $x_i(v_i,v_{-i})$ is monotone.
\item  For all $i,v_i,v_{-i}$ we have
\[ p_i(v_i,v_{-i}) = \left\{ \begin{array}{ll}
         t_i(v_{-i}) & \mbox{if $x_i(v_i,v_{-i}) = 1$};\\
        0 & \mbox{if $x_i(v_i,v_{-i})= 0$}.\end{array} \right. \] 
\end{enumerate}
\end{thm}

\section{The Geometry of Optimal Auctions}\label{sec:geometry}
Here we focus on the two-bidder case, and provide an alternative geometric interpretation of the mechanism design problem as a space partitioning problem. Our characterization holds for any number of bidders, with the appropriate generalizations and modifications; however we only address the multi-dimensional case in  Section~\ref{sec:inapproximability}, where we will use our geometric characterization to establish the inapproximability of the problem for 3 or more bidders.

We start by noting that the allocation function can be described in terms of a partition of the unit square (the space of all possible valuation pairs) into three regions:  In region $A$ bidder 1 gets the item, in region $B$ bidder 2 gets the item and in region $C$ neither gets the item. The shape of those regions is restricted by monotonicity as follows (see Figure \ref{figure1}):  Region $A$ is {\em rightward closed,} meaning that $(x,y)\in A$ and $x'>x$ implies $(x',y)\in A$, while region $B$ is, analogously, upward closed. These regions are captured by their boundaries:  Region $A$'s boundary is a function $\alpha: [0,1]\mapsto [0,1]$ where for all $y\in[0,1]$ $\alpha(y)= \inf \{x: (x,y) \in A\}$, and similarly for region $B$ and its boundary $\beta(x)$.  

\begin{figure}[h]
\centering
\includegraphics[scale = 0.4]{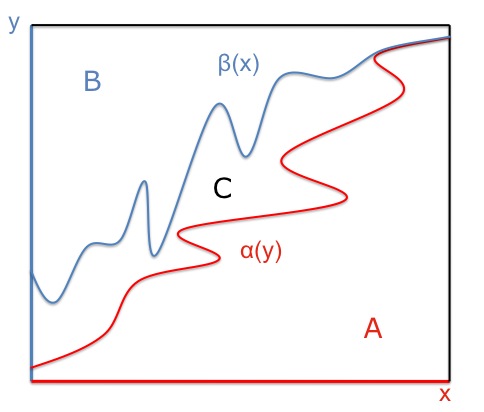}
\caption{A pair of valid allocation rules.}
\label{figure1}
\end{figure}

\begin{defn}
A {\bf valid allocation pair} $(\alpha,\beta)$ is a pair of functions from $[0,1]$ to itself satisfying the {\em non-crossing property:} for all points $(x,y)\in[0,1]^2$ we have $y\geq\beta(x)\Rightarrow x\leq\alpha(y)$.
\end{defn}

Notice that it is not necessary for the functions to be monotone; the monotonicity property of the allocation is ensured through the fact that $\alpha$ and $\beta$ are proper functions and therefore regions $A$ and $B$ are rightward and upward closed respectively. The non-crossing property ensures that for any bid pair $(x,y)$ at most one bidder is allocated the item. 

In our proof we will make extensive use of the following notion of marginal profit.

\begin{defn}\label{def:mpc}
Let $f$ (respectively, $g$) be the {\em marginal profit contribution} of a bid pair $(x,y)$ for player 1 (resp. 2) defined as:
$$f(x,y)=-\frac{\partial}{\partial x}\left[ \max_{x'\geq x}x'\cdot\int_{x'}^1\phi(t,y)\,dt\right],
g(x,y)=-\frac{\partial}{\partial y}\left[ \max_{y'\geq y}y'\cdot\int_{y'}^1\phi(x,t)\,dt\right]$$
wherever the derivative is defined, and is extended to the full range by right continuity.
\end{defn}

Intuitively, $f(x,y)\,dx\,dy$ is the added expected profit obtained from including the infinitesimal area $dx\,dy$ to $A$, that is, deciding to give the item to the first bidder if the valuations are $x,y$; in Remark~\ref{rem:discussion} we further discuss the intuition behind these functions and their relation to Myerson's virtual valuation functions.   

\begin{defn}  Call a valid allocation pair {\bf proper} if it satisfies the following condition:

\begin{equation}\label{tech-condition}
 \left.\begin{aligned}
	 \alpha(y)\cdot\int_{\alpha(y)}^1\phi(t,y)\,dt = \max_{x'\geq \alpha(y)} x'\cdot\int_{x'}^1\phi(t,y)\,dt,~~~\text{for all }y \\
	\beta(x)\cdot\int_{\beta(x)}^1\phi(x,t)\,dt  = \max_{y'\geq \beta(x)} y'\cdot\int_{y'}^1\phi(x,t)\,dt,~~~\text{for all }x 
       \end{aligned}
 \right\}
 \qquad 
\end{equation}

\end{defn}

Marginal profit contribution functions provide us with an alternative way to express the objective of expected profit.

\begin{lem}\label{lem:equiv}
Let $(\alpha,\beta)$ be a proper valid allocation pair.  Then the expected profit of any auction with payments defined as in Theorem \ref{thm:characterization}  is:
$$ \int_0^1\int_{\alpha(y)}^1f(x,y)\,dx\,dy+\int_0^1\int_{\beta(x)}^1g(x,y)\,dy\,dx$$
\end{lem}

\begin{proof}
Let $p_1(x,y), p_2(x,y)$ be the payment functions induced by the allocation rule $(\alpha,\beta)$ according to Theorem \ref{thm:characterization}. Then the expected profit of our auction is:

\begin{eqnarray*}
& & \int_0^1\int_0^1p_1(x,y)\phi(x,y)\,dx\,dy+\int_0^1\int_0^1p_2(x,y)\phi(x,y)\,dx\,dy\\
&=&\int_0^1\left[\alpha(y)\cdot\int_{\alpha(y)}^1 \phi(t,y)\,dt\right]\,dy+\int_0^1\left[\beta(x)\cdot\int_{\beta(x)}^1 \phi(x,t)\,dt\right]\,dx \\
&=&\int_0^1\left[\max_{x'\geq \alpha(y)}x'\cdot\int_{x'}^1\phi(t,y)\,dt\right]\,dy+\int_0^1\left[\max_{y'\geq \beta(x)}y'\cdot\int_{y'}^1\phi(x,t)\,dt\right]\,dx\\
&=&\int_0^1\int_{\alpha(y)}^1f(x,y)\,dx\,dy+\int_0^1\int_{\beta(x)}^1g(x,y)\,dy\,dx
\end{eqnarray*}
where in the first equality we used the characterization of truthful payments as every player's critical value, in the second equality we made use of condition~\eqref{tech-condition} and in the last equality we made use of the definition of marginal profit contribution functions.
\end{proof}

\begin{rem}\label{rem:discussion}  There is an intuitive connection between the marginal profit contribution and Myerson's (ironed) virtual valuations~\cite{RePEc:nwu:cmsems:362}.  As pointed out in the introduction, even though these functions are not identical, the goal in both cases is the same:  to capture some notion of marginal revenue. Indeed the quantity $x\cdot\int_{x}^1\phi(t,y)\,dt$ corresponds to the expected profit of the auctioneer from an agent when she is offered a price of $x$, keeping the value of the other player fixed at $y$; equivalently one can write the expected revenue as a function $R(q)$ of the probability of sale $q=1-\int_x^1\phi(t,y)\,dt$ (for a thorough discussion of this maneuver the reader is referred to~\cite{lecnotes}). Myerson's virtual valuations then correspond to the derivative of the function $R(q)$, and ironing corresponds to taking the convex hull $\hat{R}(q)$ of $R(q)$ and then differentiating. The ``ironing'' in our case corresponds to the action of taking the derivative of $R'(q)=\max_{q'\leq q} R(q')$ instead of $R(q)$, which is intuitively the ``skyline'' one would see from $q=0$. In Figure~\ref{figure12} we show an example of a revenue curve $R(q)$ and its corresponding $\hat{R}(q)$ and $R'(q)$ curves.
 
It has been already noted by Myerson~\cite{RePEc:nwu:cmsems:362} that the allocation rule will remain invariant across ironed regions: in his case this is a {\em consequence} of his proof technique involving ironing. In our case we {\em explicitly demand} that the allocation is invariant by imposing Condition \ref{tech-condition}. Since the invariance of the allocation rule across ironed regions follows from Myerson's analysis as well, Condition \ref{tech-condition} is indeed not a loss of generality in the sense that all optimal auctions satisfy it.  We prove this formally next.  

\end{rem}

The next lemma establishes that without loss of generality we can restrict ourselves to proper allocation pairs.  Let $Profit(\alpha,\beta,\phi)$ denote the profit of an auction with allocation curves $(\alpha,\beta)$, when the joint distribution of valuations is $\phi$. 

\begin{lem}\label{lem:mpc-cont}
For any $\phi$ and for any valid allocation pair $(\alpha,\beta)$ there is a proper valid pair $(\alpha',\beta')$ such that $Profit(\alpha',\beta',\phi) \geq Profit(\alpha,\beta,\phi) $ 
\end{lem}

\begin{proof}
For the sake of contradiction suppose this is not true, i.e. $Profit(\alpha',\beta',\phi) < Profit(\alpha,\beta,\phi)$ for any proper valid allocation pair $(\alpha',\beta')$.

We start by defining the following sets of points:
 $$\mathcal{Y} = \left\{y\in[0,1]\ \middle\vert\  \alpha(y)\cdot\int_{\alpha(y)}^1\phi(t,y)\,dt \neq \max_{x'\geq \alpha(y)} x'\cdot\int_{x'}^1\phi(t,y)\,dt\right\}$$
 $$\mathcal{X} = \left\{x\in[0,1]\ \middle\vert\ \beta(x)\cdot\int_{\beta(x)}^1\phi(x,t)\,dt  \neq \max_{y'\geq \beta(x)} y'\cdot\int_{y'}^1\phi(x,t)\,dt\right\} $$
 as the set of all coordinates $y$ (resp. $x$) where condition~\eqref{tech-condition} is violated by function $\alpha$ (resp. $\beta$). Consider now the auction defined by the following allocation curves:
 
 \[ \alpha'(y) = \left\{ \begin{array}{ll}
         \alpha(y) & \mbox{, if $y\notin \mathcal{Y}$};\\
         \arg\max_{x'} \{x'\int_{x'}^1\phi(t,y)\,dt\mid x'\geq \alpha(y)\} & \mbox{, if $y\in \mathcal{Y}$}.\end{array} \right. \] 
         
\[ \beta'(x) = \left\{ \begin{array}{ll}
         \beta(x) & \mbox{, if $x\notin \mathcal{X}$};\\
         \arg\max_{y'} \{y'\int_{y'}^1\phi(x,t)\,dt\mid y'\geq \beta(x)\} & \mbox{, if $x\in \mathcal{X}$}.\end{array} \right. \] 

where --in the case of ties-- $\arg\{\cdot\}$ returns the largest $y$ or $x$ respectively. By construction, the new pair $(\alpha',\beta')$ satisfies condition~\eqref{tech-condition}. In what follows we claim that the resulting allocation pair $(\alpha',\beta')$ is also valid and moreover it has greater revenue than $(\alpha,\beta)$, thus reaching a contradiction.

The monotonicity property of the allocation is satisfied since $\alpha'$ and $\beta'$ are proper functions of $y$ and $x$ respectively. The non-crossing property follows from the non-crossing property of $\alpha$ and $\beta$ and the fact that $\alpha'(y)\geq \alpha(y)$ for all $y\in[0,1]$ and $\beta'(x)\geq \beta(x)$ for all $x\in[0,1]$. Finally, for the profit of the two auctions defined by the allocation curves $(\alpha,\beta)$ and $(\alpha',\beta')$ we have:
\begin{eqnarray*}
&&Profit(\alpha',\beta',\phi)\\
&=&\int_0^1\left[\alpha'(y)\cdot\int_{\alpha'(y)}^1 \phi(t,y)\,dt\right]\,dy+\int_0^1\left[\beta'(x)\cdot\int_{\beta'(x)}^1 \phi(x,t)\,dt\right]\,dx \\
&\geq&\int_0^1\left[\alpha(y)\cdot\int_{\alpha(y)}^1 \phi(t,y)\,dt\right]\,dy+\int_0^1\left[\beta(x)\cdot\int_{\beta(x)}^1 \phi(x,t)\,dt\right]\,dx \\
&=& Profit(\alpha,\beta,\phi),
\end{eqnarray*}
where the inequality follows from the definition of $\alpha'$ and $\beta'$. The lemma now follows. 
\end{proof}

Denote now by $\cal AB$ the set of all proper valid allocations $(\alpha,\beta)$.
The problem of finding the optimal auction can now be restated as the following variational calculus-type problem:

\begin{defn}{\bf [Problem A]}
\begin{eqnarray*}
&\sup_{(\alpha,\beta)\in{\cal AB}}&\left\{\int_0^1\int_{\alpha(y)}^1f(x,y)\,dx\,dy+\int_0^1\int_{\beta(x)}^1g(x,y)\,dy\,dx\right\}\\
\end{eqnarray*}
\end{defn}

\begin{figure}[h]
\centering
\includegraphics[scale = 0.5]{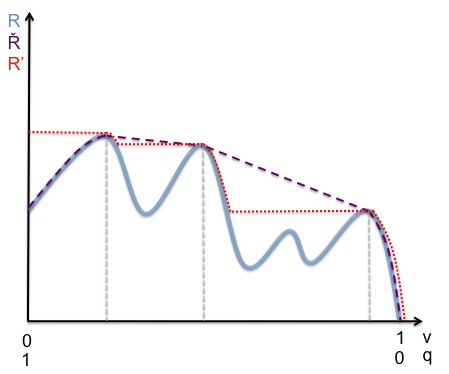}
\caption{The revenue curves $R$, $\hat{R}$ and $R'$.}
\label{figure12}
\end{figure}

\section{Two Bidders:  The Discrete Case}
In this section we present an algorithm for computing the optimal mechanism when there are two bidders with a {\em discrete} joint distribution $\phi$ with support $Sup(\phi)=[N]\times[N]$.  

The marginal profit contribution functions defined in the previous section can be appropriately modified for the discrete setting in hand:
\begin{defn}

The discrete analogues of the marginal profit contribution functions (Definition \ref{def:mpc}) for each player are defined as follows: 
$$f(i,j) =\max\left\{ i\cdot\sum_{i'\geq i}\phi(i',j) - \sum_{i'>i} f(i',j),0\right\}~~\text{for player 1}$$
$$g(i,j) =\max\left\{ j\cdot\sum_{j'\geq j}\phi(i,j') - \sum_{j'> j} g(i,j'),0\right\}~~\text{for player 2}$$
\end{defn}

As in the continuous case, we will represent the auction through a pair of functions $(\alpha,\beta)$:

\begin{defn}
A {\bf valid allocation pair} $(\alpha,\beta)$ for the discrete setting is a pair of functions from $[N]$ to itself satisfying the {\em non-crossing property:} for all points $(i,j)\in[N]\times[N]$ we have $j\geq\beta(i)\Rightarrow i<\alpha(j)$.   Such a pair partitions the set $Sup(\phi)$ into three sets $A,B$, and $C$, where $A=\{(i,j)\in Sup(\phi):  i \geq \alpha(j)\}$, $B=\{(i,j)\in Sup(\phi):  j \geq \beta(i)\}$, and $C=Sup(\phi)-A-B$.   We say that $(\alpha,\beta)$ {\bf induces} the partition $(A,B)$, where $C$ is implicit.   

Suppose that for a valid pair  $(\alpha,\beta)$ representing an optimal auction and $j\in [n]$, we have that $f(\alpha(j), j)=0$.   Then the auction represented by the same valid pair, except that $\alpha(j)$ is increased by one, is also an optimal auction (since it entails the same set of positive marginal profit contributions).  We can therefore consider, without loss of generality, only valid pairs with $f(\alpha(j),j), g(i,\beta(i))>0$.  We call such valid pairs {\em proper}.  Finally, we define $\Pi(\phi)$ to be the set of all partitions $A,B$ of $Sup(\phi)$ induced by proper valid pairs.

\end{defn}

Now, looking back at the formulation of the optimum-revenue auction problem in Problem A at the conclusion of the last section, it is immediate that obtaining the optimum-revenue auction in the discrete case is tantamount to solving a discrete optimization problem:
$$\max_{(A,B)\in\Pi(\phi)} \sum_{(i,j)\in A} f(i,j) + \sum_{(i,j)\in B} g(i,j).$$

All that remains now is to provide a useful characterization of the set $\Pi(\phi)$.  To this end, notice first that any such pair of sets $(A,B)$ has the following two additional properties, the first one inherited from incentive compatibility, and the second one following from the fact that the allocation curves form a proper valid pair:

\begin{defn}
We call a pair of disjoint subsets $(A,B)$ of $Sup(\phi)$ {\em monotone} if the following holds:
\begin{itemize}
\item If $(i,j)\in A$ and $(i',j)\in Sup(\phi)$ with $i'>i$, then $(i',j)\in A$.
\item If $(i,j)\in B$ and $(i,j')\in Sup(\phi)$ with $j'>j$, then $(i,j')\in B$.
\end{itemize}
We call such a pair {\em proper} if 
\begin{itemize}
\item If $(i,j)\in A$ and $(i',j)\notin A$ for all $i'<i$, then $f(i,j)>0$.
\item If $(i,j)\in B$ and $(i,j')\notin B$ for all $j'<j$, then $g(i,j)>0$.
That is, in proper partitions, all lower boundary points of the regions $A$ and $B$ have positive marginal profit contributions (the intuition being that otherwise, either this is not an optimal auction in the case of a negative marginal contribution, or there is another optimal auction with this property in the case of  a zero marginal contribution).
\end{itemize}

\end{defn}

Let us define now a bipartite graph $G^{\phi} = (U,W,E)$:  

\begin{itemize}
\item $U=\{u_{i,j}\mid i,j\in [N]\}$ and $W=\{w_{i,j}\mid i,j\in [N]\}$;
\item $(u_{i,j},w_{i',j'}) \in E$ if and only if $i \leq i'$ and $j \geq j'$.  In other words, there is an edge between $u$ and $w$ if and only if, informally, $(i',j')$ ``lies to the southeast'' of $(i,j)$ (see Figure~\ref{figure13}); (Notice that, in our informal sense, a point ``lies to the southeast'' of all points to its north and to its west, including the point itself.)
\item The weight of any node $u_{i,j}$ is $f(i,j)$ and the weight of any node $w_{i,j}$ is $g(i,j)$.
\end{itemize}

Intuitively, the bipartite graph captures impossibilities in constructing the optimal auction:  an edge $(u,w)$ signifies that it is not possible that both $u\in A$ and $w\in B$ (slightly abusing notation).

\begin{figure}[h]
\centering
\includegraphics[scale = 0.5]{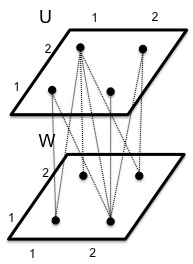}
\caption{The bipartite graph for 2 players with $N=2$.}
\label{figure13}
\end{figure}

We can now prove the sought combinatorial characterization of $\Pi(\phi)$:

\begin{lem}\label{thm:equiv-is-Pi}
Let $(A,B)$ be a pair of disjoint  subsets of $Sup(\phi)$. Then $(A,B)\in \Pi(\phi)$ if and only if $(A,B)$ is monotone and proper, and $\{u_{i,j}: (i,j)\in A\} \cup \{w_{i,j}: (i,j)\in B\}$ is an independent set of $G^{\phi}$.
\end{lem}

\begin{proof}  ({\bf If.}) Since $\{u_{i,j}: (i,j)\in A\} \cup \{w_{i,j}: (i,j)\in B\}$ is an independent set of $G^{\phi}$ it follows that the sets $A$ and $B$ are disjoint. Now, since $(A,B)$ is a pair of monotone, proper and disjoint subsets of $Sup(\phi)$ the following pair of functions $(\alpha(j),\beta(i))$ is a proper valid pair, immediately implying that $(A,B)\in \Pi(\phi)$: $\alpha(j) = \min\{i\mid(i,j)\in A\}$ and $\beta(i) = \min\{j\mid(i,j)\in B\}$.

({\bf Only if.}) If $(A,B)\in \Pi(\phi)$ then by the definition of $\Pi(\phi)$ the sets $(A,B)$ have to be monotone and proper and they also need to form a partition, i.e. be disjoint. To show that $\{u_{i,j}: (i,j)\in A\} \cup \{w_{i,j}: (i,j)\in B\}$ is an independent set of $G^{\phi}$ assume towards contradiction that there are nodes $u_{i,j}$ and $w_{i',j'}$ such that $(i,j)\in A$ and $(i',j')\in B$, with an edge between $u_{i,j}$ and $w_{i',j'}$; from  the construction of $G^\phi$, it follows that $i \leq i'$ and $j \geq j'$. Since $A$ and $B$ are both monotone, it follows that $(i',j)\in A\cap B$, contradicting the disjointness of $A$ and $B$. 
\end{proof}

The next Lemma shows that the additional assumptions that $(A,B)$ is a proper and monotone pair of subsets are not necessary, if one restricts attention to the optimal solution (i.e. the solution of maximum weight).

\begin{lem}\label{thm:monotone-proper-follows}
Let $(A,B)$ be a pair of subsets of $Sup(\phi)$ such that the set $\{u_{i,j}: (i,j)\in A\} \cup \{w_{i,j}: (i,j)\in B\}$ is a maximum weight independent set of $G^{\phi}$, of minimum cardinality among all independent sets of the same weight. Then $(A,B)$ is monotone and proper, i.e. $(A,B)\in\Pi(\phi)$.
\end{lem}

\begin{proof}
It suffices to show that the set $(A,B)$ is monotone and proper, and the lemma follows from Lemma~\ref{thm:equiv-is-Pi}. Indeed, the monotonicity of  $(A,B)$ follows from the fact that the set $\{u_{i,j}: (i,j)\in A\} \cup \{w_{i,j}: (i,j)\in B\}$ is a maximum weight independent set, and all weights are non-negative. Moreover, since the independent set has minimum cardinality among all independent sets of the same weight, it follows that it does not contain any node $u_{i,j}$ of zero weight, i.e. corresponding to some valuation $(i,j)$ such that $f(i,j) = 0$, unless it also includes a node $u_{i',j}$ corresponding to some valuation $(i',j)$ with $f(i',j) > 0$ for some $i'<i$ (and analogously for $w_{i,j}$). Hence, by definition, $(A,B)$ is proper as well.
\end{proof}

\begin{thm}
Given a discrete joint valuation distribution $\phi$ for two bidders, the optimal ex-post IC and IR deterministic auction can be computed in time $O(|Sup(\phi)|^3)$.
\end{thm}

\begin{proof}
It follows from Lemma~\ref{thm:monotone-proper-follows} that computing the optimal auction for two bidders with a joint valuation distribution $\phi$ reduces to computing a maximum weight independent set on the induced bipartite graph $G^{\phi}$. In particular the optimal allocation rule corresponds to a partition $(A,B)$ such that $\{u_{i,j}: (i,j)\in A\} \cup \{w_{i,j}: (i,j)\in B\}$ is a maximum weight independent set of $G^{\phi}$, with minimum cardinality among all independent sets of the same weight. Finding the maximum weight independent set by running a min-cost-flow algorithm yields the desired running time.
\end{proof}

\section{Two Bidders:  The General Case}\label{sec:continuous}
In this section we return to the continuous two-bidder problem of Section 3. Our main result is an efficient algorithm that approximates the optimum solution of Problem A within an additive $\epsilon$.   Our main tool is a duality theorem, generalizing the duality between the maximum-weight  independent set problem in a bipartite graph and a minimum-cost flow in an associate network.  In particular, we show that the maximization Problem A defined in Section 3 is equivalent to a certain mass-moving minimization problem (reminiscent in some aspects of the classic Monge-Kantorovich \cite{Evans99partialdifferential} mass-transfer problem), namely the following:

\begin{defn}{\bf [Problem B]}
\begin{eqnarray*}
&\inf_\gamma&\left\{\int_0^1\int_0^1\int_0^{y_1}\int_{x_1}^1\gamma(x_1,y_1,x_2,y_2)\,dx_2\, dy_2\,dx_1\,dy_1\right\}\\
&\mathrm{s.t.}& \int_0^{y_1}\int_{x_1}^1\gamma(x_1,y_1,x_2,y_2)\,dx_2\, dy_2 \geq f(x_1,y_1),\,\forall(x_1,y_1)\in[0,1]^2\\
&& \int_{y_2}^1\int_0^{x_2}\gamma(x_1,y_1,x_2,y_2)\,dx_1\, dy_1 \geq g(x_2,y_2),\,\forall(x_2,y_2)\in[0,1]^2\\
&& \gamma(x_1,y_1,x_2,y_2)\geq 0,\,\forall(x_1,y_1),(x_2,y_2)\in[0,1]^2\\
\end{eqnarray*}
\end{defn}

Let us denote by $\Gamma$ the set of all functions $\gamma:[0,1]^4\mapsto \Re$ satisfying the above constraints.  

One intuitive interpretation of Problem B is this:  We are given two landscapes in the unit square, captured by two functions $f,g:[0,1]^2\mapsto \Re$.  We are seeking a plan for transforming landscape $f$ to landscape $g$, where the following operations are allowed:
\begin{itemize}
\item take material away from any point $(x,y)$;
\item add material to any point $(x,y)$;
\item transfer material from any point to any other point {\em in the southeast direction} (if some material is not moved, we think of it as having moved in the southeast direction zero distance).
\end{itemize}
We want the plan in which the total amount of material moved (irrespective of distance moved, here is where this problem differs significantly from Monge-Kantorovich) is minimized.  We next show that this problem coincides, at optimality, with the optimal auction:

\begin{thm}\label{thm-main}{\bf [Duality Theorem]\\}
For any joint density function $\phi$ on $[0,1]^2$:
\begin{eqnarray*}
&&\sup_{(\alpha,\beta)\in {\cal AB}}\left\{\int_0^1\int_{\alpha(y)}^1f(x,y)\,dx\,dy+\int_0^1\int_{\beta(x)}^1g(x,y)\,dy\,dx\right\}\\
& =& \inf_{\gamma\in \Gamma}\left\{\int_0^1\int_0^1\int_0^{y_1}\int_{x_1}^1\gamma(x_1,y_1,x_2,y_2)\,dx_2\, dy_2\,dx_1\,dy_1\right\} 
\end{eqnarray*}
\end{thm}
 
\subsection{Proof of the Duality Theorem}\label{proof_main}
\paragraph{General Plan.}  The proof of the theorem is by discretizing the unit square into domains of small size, proving a duality result for the discrete version, establishing upper bounds for the discretization error, and taking the limit for finer and finer discretization. In the course of the proof we will introduce a number of auxiliary problems.

\paragraph{Discretization.} We start by discretizing the continuous functions $f$ and $g$ defined on $[0,1]\times[0,1]$ by two discrete functions $f_d$ and $g_d$ defined on the $[n] \times [n]$ grid, where $n$ is an integer greater than one and [n]=\{0,1,\ldots,n-1\}; we denote $1\over n$ by $\epsilon$.  We subdivide the $[0,1]\times[0,1]$ square into $\epsilon\times\epsilon$ little squares; we are mapping a little square with southwest coordinate $(x,y)$ to the grid point $n\cdot(x,y)$. The discrete functions are now obtained by assigning to each point in the grid the aggregate mass of its corresponding square on the plane.

$$f_d(i,j) = \int_{\epsilon i}^{\epsilon (i+1)}\int_{\epsilon j}^{\epsilon (j+1)} f(x,y)\,dy\,dx,\,\text{ for }i,j=0,\ldots,n-1$$
and
$$g_d(i,j) = \int_{\epsilon i}^{\epsilon (i+1)}\int_{\epsilon j}^{\epsilon (j+1)} g(x,y)\,dy\,dx\,\text{ for }i,j=0,\ldots,n-1.$$

\begin{figure}[h]
\centering
\includegraphics[scale = 0.5]{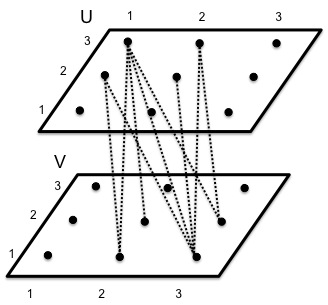}
\caption{The graph consisting of two grids, one for each player, for $n=3$.}
\label{figure2}
\end{figure}

\paragraph{The Graph.}  We next create a weighted bipartite graph $G=(U,V,E)$ as follows:  $U=V=[n]\times [n]$.  We use $u(i,j)$ and $v(i,j)$ to denote the vertices of $U$ and $V$ respectively, and sometimes use the shorthand $u$ and $v$ to refer to nodes of each grid respectively. Vertex $u(i,j)$ of $U$ has a weight equal to $w_u=f_d(i,j)$, and similarly vertex $v(i',j')$ of $V$ has a weight $w_v=g_d(i',j')$.   A pair $(u(i,j),v(i',j'))$ is in $E$ if and only if $i< i'$ and $j > j'$, that is, if grid point $v$ is in the (strictly) Southeast direction from grid point $u$ (see Figure \ref{figure2} for an example with $n=3$).

Consider now the following problem, familiar from the previous section:

\begin{defn}{\bf [Problem C: {\sc Maximum Weight Independent Set}]}\\
Given the weighted bipartite graph $G=(U,V,E)$ above,
\begin{eqnarray*}
&\max_{\{x_u\in\{0,1\},x_v\in\{0,1\}\}}&\sum_{u\in U,v\in V} x_uw_u+x_vw_v\\
&\mathrm{s.t.}&x_u+x_v \leq 1,\,\forall (u,v)\in E\\
\end{eqnarray*}
\end{defn}

The dual of the above problem is the following:
\begin{defn}{\bf [Problem D: {\sc Minimum Cost Transshipment}]}\\
Given the weighted bipartite graph $G=(U,V,E)$ above,
\begin{eqnarray*}
&\min_{\{y_{uv}\in\Re\}}&\sum_{(u,v)\in E} y_{uv}\\
&\mathrm{s.t.}&\sum_{v:(u,v)\in E} y_{uv}\geq w_u,\,\forall u\in U\\
&&\sum_{u:(u,v)\in E} y_{uv}\geq w_v,\,\forall v\in V\\
&&y_{uv}\geq 0,\,\forall (u,v)\in E\\
\end{eqnarray*}
\end{defn}

\paragraph{The Inequalities.} The crux of the proof is a sequence of results relating the various solutions and optimum solutions of these four problems.  In what follows we use $SOL(\cdot)$ to denote the cost of any feasible solution of a problem among (A), (B), (C), and (D) defined above, and $OPT(\cdot)$ to denote the cost of the optimum solution of a problem (sometimes $SOL$ and $OPT$ also denotes the actual solutions). The first such inequality establishes a form of weak duality between Problems $A$ and $B$, while the next two show that the discretization error is small.

\begin{lem}\label{lem:AB}
For any feasible solutions of $A$ and $B$ we have: $SOL(A)\leq SOL(B)$.
\end{lem}
\begin{proof}
\begin{eqnarray*}
&&SOL(A)\\
&=&\int_0^1\int_{\alpha(y_1)}^1f(x_1,y_1)\,dx_1\,dy_1+\int_0^1\int_{\beta(x_2)}^1g(x_2,y_2)\,dy_2\,dx_2\\
&\leq&\int_0^1\int_{\alpha(y_1)}^1\int_0^{y_1}\int_{x_1}^1\gamma(x_1,y_1,x_2,y_2)\,dx_2\,dy_2\,dx_1\,dy_1+\int_0^1\int_{\beta(x_2)}^1\int_{y_2}^1\int_0^{x_2}\gamma(x_1,y_1,x_2,y_2)\,dx_1\,dy_1\,dy_2\,dx_2
\end{eqnarray*}
where we used the inequality constraints of Problem B to upper bound the values of $f(x,y)$ and $g(x,y)$. We next notice that:
$$\int_0^1\int_{\beta(x_2)}^1\int_{y_2}^1\int_0^{x_2}\gamma(x_1,y_1,x_2,y_2)\,dx_1\,dy_1\,dy_2\,dx_2\leq\int_0^1\int_0^{\alpha(y_1)}\int_0^{y_1}\int_{x_1}^1\gamma(x_1,y_1,x_2,y_2)\,dx_2\,dy_2\,dx_1\,dy_1$$
This inequality follows from the non-negativity of $\gamma$ and the fact that the $(x_1,y_1,x_2,y_2)$ included in the integral of the LHS are the following set:
\begin{eqnarray*}
&&\{(x_1,y_1,x_2,y_2)\in[0,1]^4\mid y_2\geq \beta(x_2),\ x_1\leq x_2,\ y_1\geq y_2\}\\
&=&\{(x_1,y_1,x_2,y_2)\in[0,1]^4\mid y_2\geq \beta(x_2),\ x_1\leq x_2,\ y_1\geq y_2,\ x_1\leq \alpha(y_1)\}\\
&\subseteq&\{(x_1,y_1,x_2,y_2)\in[0,1]^4\mid x_2\leq \alpha(y_2),\ x_1\leq x_2,\ y_1\geq y_2,\ x_1\leq \alpha(y_1)\}\\
&\subseteq&\{(x_1,y_1,x_2,y_2)\in[0,1]^4\mid x_1\leq x_2,\ y_1\geq y_2,\ x_1\leq \alpha(y_1)\}
\end{eqnarray*}
which is exactly the set of $(x_1,y_1,x_2,y_2)$ included in the integral of the RHS. The first equality above follows from the fact that the inequality $x_1\leq \alpha(y_1)$ follows from the inequalities $\{y_2\geq \beta(x_2),\ x_1\leq x_2,\ y_1\geq y_2\}$, because we have $y_1\geq \beta(x_2)$ so the non-crossing property implies that $x_2\leq \alpha(y_1)$ and therefore $x_1\leq\alpha(y_1)$. The first set inclusion follows from the fact that $y_2\geq\beta(x_2)\Rightarrow x_2\leq\alpha(y_2)$ from the non-crossing property, while the last inclusion is trivial.

We have therefore concluded that the cost of any feasible solution of $A$ is upper bounded by:
\begin{eqnarray*}
&&\int_0^1\int_{\alpha(y_1)}^1\int_0^{y_1}\int_{x_1}^1\gamma(x_1,y_1,x_2,y_2)\,dx_2\,dy_2\,dx_1\,dy_1+\int_0^1\int_0^{\alpha(y_1)}\int_0^{y_1}\int_{x_1}^1\gamma(x_1,y_1,x_2,y_2)\,dx_2\,dy_2\,dx_1\,dy_1\\
&=&\int_0^1\int_0^1\int_0^{y_1}\int_{x_1}^1\gamma(x_1,y_1,x_2,y_2)\,dx_2\, dy_2\,dx_1\,dy_1 \\
&=& SOL(B)
\end{eqnarray*}
\end{proof}

\begin{lem}\label{lem:AC}
For the optimal solutions of $A$ and $C$ we have: $OPT(A)\geq OPT(C)-\epsilon$.
\end{lem}
\begin{proof}
Consider the optimal solution of Problem $C$; we will use it to come up with a feasible solution for Problem $A$ such that $SOL(A)\geq OPT(C)-\epsilon$. We start with the following solution: we allocate the item to player 1 for all valuations $(x,y)$ such that $x_{u(\left\lfloor\frac{x}{\epsilon}\right\rfloor,\left\lfloor\frac{y}{\epsilon}\right\rfloor)}=1$; we allocate to player 2 for all valuations $(x,y)$ such that $x_{v(\left\lfloor\frac{x}{\epsilon}\right\rfloor,\left\lfloor\frac{y}{\epsilon}\right\rfloor)}=1$, {\em and} $x_{u(\left\lfloor\frac{x}{\epsilon}\right\rfloor,\left\lfloor\frac{y}{\epsilon}\right\rfloor)}=0$; and finally, we allocate to nobody for all valuations $(x,y)$ such that $x_{u(\left\lfloor\frac{x}{\epsilon}\right\rfloor,\left\lfloor\frac{y}{\epsilon}\right\rfloor)}=0$ and $x_{v(\left\lfloor\frac{x}{\epsilon}\right\rfloor,\left\lfloor\frac{y}{\epsilon}\right\rfloor)}=0$.

We next show that the resulting allocation regions have the shape of Figure \ref{figure3}, meaning that the borders of those regions consist a valid allocation pair. First notice that for any pair of valuations $(x,y)$ --including those for which $x_{u(\left\lfloor\frac{x}{\epsilon}\right\rfloor,\left\lfloor\frac{y}{\epsilon}\right\rfloor)}=1$ and $x_{v(\left\lfloor\frac{x}{\epsilon}\right\rfloor,\left\lfloor\frac{y}{\epsilon}\right\rfloor)}=1$-- only one player gets allocated the item, so the non-crossing property is satisfied. To see why the regions are rightward and upward closed consider two nodes $u(i,j)$ and $u(i',j)$ on player 1's grid, where $i'>i$. Notice that the set of nodes on player 2's grid that node $u(i,j)$ of player 1's grid is connected to, is a strict superset of the nodes that node $u(i',j)$ is connected to. Hence, if the {\em maximum} weight independent set includes node $u(i,j)$ on the grid of player 1, it should also include $u(i',j)$ for all values $i'>i$.

This gives us two stairwise curves which --although being a valid allocation pair-- may fail to satisfy Condition~\ref{tech-condition}, and hence may not be a feasible solution for Problem $A$. To turn them into a proper valid pair, we can follow the same procedure as the one in the proof of Lemma \ref{lem:mpc-cont} and come up with a feasible solution $SOL(A)$ for Problem $A$.

Because of the aforementioned transformation the cost of this solution is {\em greater or equal} to the cost of the optimal solution $OPT(C)$ {\em minus} the contribution to the weight of the independent set by those nodes $v(i,j)$ for which the corresponding node $u(i,j)$ on the grid of player 1 is also included in the independent set. The reason for that is that for valuations $(x,y)$ such that $x_{v(\left\lfloor\frac{x}{\epsilon}\right\rfloor,\left\lfloor\frac{y}{\epsilon}\right\rfloor)}=1$ {\em and} $x_{u(\left\lfloor\frac{x}{\epsilon}\right\rfloor,\left\lfloor\frac{y}{\epsilon}\right\rfloor)}=1$, our solution explicitly allocates the item only to player 1, therefore losing the weight contribution of node $v(\left\lfloor\frac{x}{\epsilon}\right\rfloor,\left\lfloor\frac{y}{\epsilon}\right\rfloor)$. In what follows we argue that this results in the loss of an $\epsilon$-additive factor, so that the cost of the resulting solution is at least:
$$\sum_{u\in U}x_u f_d(u)+\sum_{v\in V}x_v g_d(v) -\epsilon= OPT(C) -\epsilon$$
To show this we first argue that the number of nodes $v(i,j)$ for which this happens is small, in particular there can only be at most $1/\epsilon$ such nodes. To see this notice that in the constructed feasible solution to Problem $A$, these nodes lie on the boundary between regions where player 1 gets the item and player 2 gets the item; any such boundary has to be monotone, since it corresponds to the overlap of the two allocation curves $\alpha,\beta$, and it can therefore contain at most $1/\epsilon$ nodes. Next notice that the value of $g$ at any point $(x,y)$ is at most 1: indeed, $g(x,y)$ is defined as $y\phi(x,y) - \int_y^1\phi(x,t)dt$, wherever $\frac{\partial}{\partial y}\left[ \max_{y'\geq y}y'\cdot\int_{y'}^1\phi(x,t)\,dt\right]$ is defined, and extended to full range by right continuity. It follows immediately that $g(x,y)\leq 1$ and $w_{v(i,j)}=g_d(i,j)\leq \epsilon^2$; therefore the total weight loss is at most $\frac{1}{\epsilon}\cdot \epsilon^2$ and the lemma follows.
\end{proof}

\begin{figure}[h]
\centering
\includegraphics[scale = 0.4]{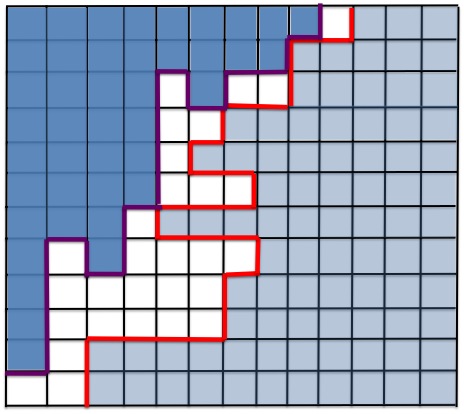}
\caption{The solution for Problem $A$ prior to the transformation of Lemma~\ref{lem:mpc-cont}.}
\label{figure3}
\end{figure}

\begin{lem}\label{lem:DB}
For the optimal solutions of $B$ and $D$ we have:  $OPT(D)\geq OPT(B)$.
\end{lem}
\begin{proof}
Given a feasible solution for Problem $D$, we will come up with a feasible solution of the same cost for Problem $B$. The optimum solution for Problem $B$ will have at most that cost and the lemma follows.

We start by defining $\gamma(x_1,y_1,x_2,y_2) $, for any pair of points $(x_1,y_1), (x_2,y_2)$ where the $\epsilon^2$-area square containing $(x_2,y_2)$ lies in the (strict) southeast orthant of the $\epsilon^2$-area square containing $(x_1,y_1)$, as follows:

$$\gamma(x_1,y_1,x_2,y_2)= f(x_1,y_1)\cdot g(x_2,y_2)\cdot\frac{y_{uv}}{w_uw_v},~~~~~\text{if}~~\left\lfloor\frac{x_1}{\epsilon}\right\rfloor< \left\lfloor\frac{x_2}{\epsilon}\right\rfloor~~\text{and}~~\left\lfloor\frac{y_1}{\epsilon}\right\rfloor> \left\lfloor\frac{y_2}{\epsilon}\right\rfloor$$

 where $u\in U$ (resp. $v\in V$) is the grid point $u\left(\left\lfloor\frac{x_1}{\epsilon}\right\rfloor,\left\lfloor\frac{y_1}{\epsilon}\right\rfloor\right)$ (resp. $v\left(\left\lfloor\frac{x_2}{\epsilon}\right\rfloor,\left\lfloor\frac{y_2}{\epsilon}\right\rfloor\right)$) that corresponds to the little $\epsilon^2$-area square containing point $(x_1,y_1)$ (resp. $(x_2,y_2))$. Finally, we let:
 
 $$\gamma(x_1,y_1,x_2,y_2)=0 ,~~~~~\text{if}~~\left\lfloor\frac{x_1}{\epsilon}\right\rfloor\geq \left\lfloor\frac{x_2}{\epsilon}\right\rfloor~~\text{or}~~\left\lfloor\frac{y_1}{\epsilon}\right\rfloor\leq\left\lfloor\frac{y_2}{\epsilon}\right\rfloor$$
 
We next verify that the function $\gamma$ defined above satisfies the constraints of Problem $B$. Since the non-negativity constraint is obviously satisfied, we only need to check that $\gamma$ satisfies the first and second constraints of Problem $B$. We only provide the proof for the first constraint and the proof for the second constraint follows along the exact same lines: 

\begin{eqnarray*}
&& \int_0^{y_1}\int_{x_1}^1\gamma(x_1,y_1,x_2,y_2)\,dx_2\, dy_2\\
  &=&\sum_{i>\left\lfloor\frac{x_1}{\epsilon}\right\rfloor, j<\left\lfloor\frac{y_1}{\epsilon}\right\rfloor}\int_{\epsilon i}^{\epsilon (i+1)}\int_{\epsilon j}^{\epsilon (j+1)}f(x_1,y_1)\cdot  g(x_2,y_2)\cdot\frac{y_{u\left(\left\lfloor\frac{x_1}{\epsilon}\right\rfloor,\left\lfloor\frac{y_1}{\epsilon}\right\rfloor\right),v\left(\left\lfloor\frac{x_2}{\epsilon}\right\rfloor,\left\lfloor\frac{y_2}{\epsilon}\right\rfloor\right)}}{w_{u\left(\left\lfloor\frac{x_1}{\epsilon}\right\rfloor,\left\lfloor\frac{y_1}{\epsilon}\right\rfloor\right)}\cdot w_{v\left(\left\lfloor\frac{x_2}{\epsilon}\right\rfloor,\left\lfloor\frac{y_2}{\epsilon}\right\rfloor\right)}}\,dy_2\,dx_2\\
 &=&f(x_1,y_1)\cdot\sum_{i>\left\lfloor\frac{x_1}{\epsilon}\right\rfloor, j<\left\lfloor\frac{y_1}{\epsilon}\right\rfloor}\frac{y_{u\left(\left\lfloor\frac{x_1}{\epsilon}\right\rfloor,\left\lfloor\frac{y_1}{\epsilon}\right\rfloor\right),v(i,j)}}{w_{u\left(\left\lfloor\frac{x_1}{\epsilon}\right\rfloor,\left\lfloor\frac{y_1}{\epsilon}\right\rfloor\right)}\cdot w_{v(i,j)}}\cdot\int_{\epsilon i}^{\epsilon (i+1)}\int_{\epsilon j}^{\epsilon (j+1)}  g(x_2,y_2)\,dy_2\,dx_2\\
  &=&f(x_1,y_1)\cdot\sum_{i>\left\lfloor\frac{x_1}{\epsilon}\right\rfloor, j<\left\lfloor\frac{y_1}{\epsilon}\right\rfloor}\frac{y_{u\left(\left\lfloor\frac{x_1}{\epsilon}\right\rfloor,\left\lfloor\frac{y_1}{\epsilon}\right\rfloor\right),v(i,j)}}{w_{u\left(\left\lfloor\frac{x_1}{\epsilon}\right\rfloor,\left\lfloor\frac{y_1}{\epsilon}\right\rfloor\right)}\cdot w_{v(i,j)}}\cdot w_{v(i,j)}\\
 &=&f(x_1,y_1)\cdot\sum_{i>\left\lfloor\frac{x_1}{\epsilon}\right\rfloor, j<\left\lfloor\frac{y_1}{\epsilon}\right\rfloor}\frac{y_{u\left(\left\lfloor\frac{x_1}{\epsilon}\right\rfloor,\left\lfloor\frac{y_1}{\epsilon}\right\rfloor\right),v(i,j)}}{w_{u\left(\left\lfloor\frac{x_1}{\epsilon}\right\rfloor,\left\lfloor\frac{y_1}{\epsilon}\right\rfloor\right)}}\\
&\geq& f(x_1,y_1)
 \end{eqnarray*} 
where in the first equality we split the integration over discretized square regions of area $\epsilon^2$ (the same that are used in the discrete auxiliary Problems $C$ and $D$) and in the second equality we rearranged the order of summation and integration, noticing that the weights $w$ and flows $y$ remain constant across the discretized squares (independently of the actual value of $(x_2,y_2)$). In the third equality we used the definition of the weight $w$ and in the last inequality we used the fact that $y$ is a feasible solution for Problem $D$ and therefore $\sum_{v\in E}y_{uv}\geq w_u$.

We conclude our proof by showing that the cost of the feasible solution we produced is exactly $OPT(D)$: 
\begin{eqnarray*}
&&\int_0^1\int_0^1\int_0^{y_1}\int_{x_1}^1\gamma(x_1,y_1,x_2,y_2)\,dx_2\, dy_2\,dx_1\,dy_1\\
&=&\sum_{i,j}\int_{\epsilon i}^{\epsilon (i+1)}\int_{\epsilon j}^{\epsilon (j+1)}\int_0^{y_1}\int_{x_1}^1\gamma(x_1,y_1,x_2,y_2)\,dx_2\, dy_2\,dx_1\,dy_1\\
&=&\sum_{i,j}\int_{\epsilon i}^{\epsilon (i+1)}\int_{\epsilon j}^{\epsilon (j+1)} f(x_1,y_1)\cdot\sum_{i'>\left\lfloor\frac{x_1}{\epsilon}\right\rfloor, j'<\left\lfloor\frac{y_1}{\epsilon}\right\rfloor}\frac{y_{u\left(\left\lfloor\frac{x_1}{\epsilon}\right\rfloor,\left\lfloor\frac{y_1}{\epsilon}\right\rfloor\right),v(i',j')}}{w_{u\left(\left\lfloor\frac{x_1}{\epsilon}\right\rfloor,\left\lfloor\frac{y_1}{\epsilon}\right\rfloor\right)}}
\,dx_1\,dy_1\\
&=&\sum_{i,j}\sum_{i'>i, j'<j}\frac{y_{u(i,j),v(i',j')}}{w_{u(i,j)}}\cdot \int_{\epsilon i}^{\epsilon (i+1)}\int_{\epsilon j}^{\epsilon (j+1)} f(x_1,y_1)
\,dx_1\,dy_1\\
&=&\sum_{i,j}\sum_{i'>i, j'<j}\frac{y_{u(i,j),v(i',j')}}{w_{u(i,j)}}\cdot w_{u(i,j)}\\
&=&\sum_{i,j}\sum_{i'>i, j'<j}y_{u(i,j),v(i',j')}\\
&=& OPT(D)
\end{eqnarray*}
where in the first equality we split the integration of $(x_1,y_1)$ over discretized square regions of area $\epsilon^2$ and in the second equality we plugged in the expression for $\int_0^{y_1}\int_{x_1}^1\gamma(x_1,y_1,x_2,y_2)\,dx_2\, dy_2$ that we had derived from our previous proof establishing that the first constraint of Problem $B$ was satisfied. In the third equality we once again rearranged the order of summation and integration, noticing that the weights $w$ and flows $y$ remain constant across the discretized squares (independently of the actual value of $(x_1,y_1)$), in the fourth equality we used the definition of the weight $w$ and in the last equality we replaced with the objective function of Problem $D$.

\end{proof}

\begin{proof}[\it Proof of the Duality Theorem]
First notice that, by strong duality and since the constraint matrix of Problem C is totally unimodular, we have that $OPT(C) = OPT(D)$; combining this with Lemmas \ref{lem:AC} and \ref{lem:DB} we get $OPT(A)\geq OPT(B) - \epsilon$. From Lemma \ref{lem:AB} we get $OPT(A)\leq OPT(B)$. By having $\epsilon\rightarrow 0$ we get the result.
\end{proof}

\subsection{The Algorithm}
The proof of the Main Theorem suggests a fully polynomial-time approximation scheme (FPTAS) for the continuous case, that is, a mechanism that approximates the optimal profit within additive error $\epsilon$, and runs in time polynomial in $1\over \epsilon$.   In the algorithm and the correctness proof, we assume that the continuous joint distribution $\phi$ is Lipschitz continuous, and that it is presented through oracle access.  It is easy to see that these assumptions are essentially necessary, in that no approximation (or meaningful solution of any other nature) is possible when the function can be arbitrarily discontinuous, or is inaccessible for large parts of the domain.

\begin{algorithm}[H]
 \caption{{\sc OptimalAuction} for two bidders with continuous distributions}\label{algorithm}
 \begin{algorithmic}[1]
 \State{\bf Input:} probability distribution $\phi\in [0,1]^2$
 \State Compute $f(\cdot,\cdot),g(\cdot,\cdot)$
 \State Discretize the unit plane and construct the bipartite graph $G$ as described in Section \ref{proof_main}
    \State Compute a {\sc Maximum Weight Independent Set} for $G$ 
    \State{\bf Output: $(\alpha,\beta)$},  the valid allocation pair corresponding to the stair-like curves of Figure \ref{figure3}
 \end{algorithmic}
\end{algorithm}

The following theorem establishes that our algorithm has the desired properties.

\begin{thm}\label{thm:alg}
Algorithm \ref{algorithm} returns a truthful mechanism that approximates the optimal profit within $\epsilon$ additive error; moreover the algorithm runs in time polynomial in $1/\epsilon$.
\end{thm}
\begin{proof}
Algorithm \ref{algorithm} returns a valid allocation pair so it is truthful by construction. However the allocation pair $(\alpha,\beta)$ returned may well not satisfy condition~\eqref{tech-condition} and may consequently not constitute a feasible solution to Problem A. This is problematic since it does not allow us to use Lemma~\ref{lem:equiv} to compute the profit of the auction returned. To that end we need to establish that the violation of condition~\eqref{tech-condition} is --in some sense-- negligible; we do that next.

Suppose that curve $\alpha(y)$ violates condition~\eqref{tech-condition} for some $y^\ast\in[0,1]$ and let 
$$x^\ast = \arg\max_{x\geq \alpha(y^\ast)} x\cdot\int_{x}^1\phi(t,y^\ast)\,dt,$$
be the minimum $x$ for which we could create a new solution $\alpha'$ by setting $\alpha'(y^\ast)= x^\ast$ and have condition~\eqref{tech-condition} restored, while not altering the profit of our auction.\footnote{The reader is referred to Lemma~\ref{lem:mpc-cont} for further discussion on this point.} We will argue that
 \begin{equation}\label{ineq:correctness0}
 \alpha(y^\ast)\cdot\int_{\alpha(y^\ast)}^1\phi(t,y^\ast)\,dt \geq x^\ast\cdot\int_{x^\ast}^1\phi(t,y^\ast)\,dt - \Theta(\epsilon)
 \end{equation}
To do that we consider the node corresponding to the little square on the unit plane containing $(\alpha(y^\ast),y^\ast)$. Since this node belongs to the boundary of the allocation region of player 1, we can assume wlog that it has non-zero weight; hence there must exist some point $(x_1,y_1)$ in the corresponding square on the unit plane with $f(x_1,y_1)>0$. By the definition of $f$ this immediately implies that 
$$x_1\cdot\int_{x_1}^1\phi(t,y_1)\,dt \geq x\cdot\int_x^1\phi(t,y_1)\,dt,~~~~\text{for all }x\geq x_1$$
and therefore in particular that
\begin{equation}\label{ineq:correctness}
x_1\cdot\int_{x_1}^1\phi(t,y_1)\,dt \geq x^\ast\cdot\int_{x^\ast}^1\phi(t,y_1)\,dt.
\end{equation}
Since the $l_1$-distance of points $(\alpha(y^\ast),y^\ast)$ and $(x_1,y_1)$ and of points $(x^\ast,y^\ast)$ and $(x^\ast,y_1)$ is at most $2\epsilon$, we get that:
$$\alpha(y^\ast)\cdot\int_{\alpha(y^\ast)}^1\phi(t,y^\ast)\,dt 
\geq x_1\cdot\int_{\alpha(x_1)}^1\phi(t,y_1)\,dt - 2\lambda\epsilon
\geq x^\ast\cdot\int_{x^\ast}^1\phi(t,y_1)\,dt - 2\lambda\epsilon
\geq x^\ast\cdot\int_{x^\ast}^1\phi(t,y^\ast)\,dt - 4\lambda\epsilon$$
where in the first and third inequalities we used the fact that $x\cdot\int_x^1\phi(t,y)\,dt$ is Lipschitz-continuous for some constant $\lambda$ and in the second inequality we used inequality~\eqref{ineq:correctness}. The exact same argument applies for $\beta(x)$ as well.

We are now ready to prove a lower bound on the profit of the auction returned by our algorithm; in what follows we use $\alpha'(y)$ and $\beta'(x)$ to denote the allocation curves that would result by the aforementioned transformation. Note that by construction it holds that
\begin{equation}\label{ineq:correctness1}
\alpha'(y)\cdot\int_{\alpha'(y)}^1 \phi(x,y)\, dx = \int_{\alpha'(y)}^1 f(x,y)\, dx~~~\text{and}~~~
\beta'(x)\cdot\int_{\beta'(x)}^1 \phi(x,y)\, dy = \int_{\beta'(x)}^1 g(x,y)\, dy
\end{equation}
so the profit of the algorithm is:
\begin{eqnarray*}
&&\int_0^1\left[\alpha(y)\cdot\int_{\alpha(y)}^1 \phi(x,y)\,dx\right]\,dy+\int_0^1\left[\beta(x)\cdot\int_{\beta(x)}^1 \phi(x,y)\,dy\right]\,dx\\
&\geq&\int_0^1\left[\alpha'(y)\cdot\int_{\alpha'(y)}^1 \phi(x,y)\,dx-\Theta(\epsilon)\right]\,dy+\int_0^1\left[\beta'(x)\cdot\int_{\beta'(x)}^1 \phi(x,y)\,dy-\Theta(\epsilon)\right]\,dx\\
&=&\int_0^1\int_{\alpha'(y)}^1f(x,y)\,dx\,dy+\int_0^1\int_{\beta'(x)}^1g(x,y)\,dy\,dx - \Theta(\epsilon)\\
&=& OPT(C) - \Theta(\epsilon)\\
&\geq& OPT(A) - \Theta(\epsilon)
\end{eqnarray*}
where in the first inequality we used inequality~\eqref{ineq:correctness0}, in the first equality we used~\eqref{ineq:correctness1} and in the last inequality we used our Main Theorem from the previous section.
 
In terms of running time, the discretized approximations of $f$ and $g$ are trivial (because of Lipschitz continuity, we can take $f(i,j) = f(i\epsilon,j\epsilon)$, and similarly for $g$).  Solving the {\sc Maximum Weight Independent Set} problem is done exactly as in the previous section.
\end{proof}

\section{NP-completeness}\label{sec:inapproximability}
We show that for 3 bidders the problem of designing an approximately optimal (deterministic) auction becomes NP-hard.  The current proof establishes hardness for a small threshold around 0.05\%; we believe that this will not be too hard to improve upon.  

\subsection{A geometric characterization for three bidders}
We start by formally defining the discrete version of the problem we will prove to be NP-hard; to simplify the exposition of the problem, we assume that the support of the discrete distribution is $\mathcal{G}=[S]\times[S]\times[S]$.  We are interested in this problem:

\begin{defn}[{\sc 3OptimalAuctionDesign}]\label{defn:3optauctdesign}
Given a joint discrete probability distribution $\phi$ supported on $\mathcal{G}$, find the optimal, ex-post IC and IR deterministic auction, denoted by a 3-dimensional allocation matrix $A$, where $A[x,y,z] = i$, with $i$ being the index of the bidder who gets the item when the bid vector is $(x,y,z)$, or 0 if the auctioneer keeps the item.
\end{defn}

As was the case with 2 bidders, the following notion of marginal profit contribution, appropriately modified for the discrete case, will be useful to our proof.
\begin{defn}
The discrete analogues of the marginal profit contribution functions (Definition \ref{def:mpc}) for each player are the following: 
$$f(x,y,z) =\max\left\{ x\cdot\sum_{x'\geq x}\phi(x',y,z) - \sum_{x'> x} f(x',y,z),0\right\}~~\text{for player 1}$$
$$g(x,y,z) =\max\left\{ y\cdot\sum_{y'\geq y}\phi(x,y',z) - \sum_{y'> y} f(x,y',z),0\right\}~~\text{for player 2}$$
$$h(x,y,z) =\max\left\{ z\cdot\sum_{z'\geq z}\phi(x,y,z') - \sum_{z'> z} f(x,y,z'),0\right\}~~\text{for player 3}$$
\end{defn}

\paragraph{The Segments.} Given a distribution $\phi(x,y,z)$ over the points of $\mathcal{G}$, any node $(x,y,z)$ of $\mathcal{G}$ with $f(x,y,z)>0$ is the starting point of what we shall henceforth be calling an \emph{$x$-segment}: an interval (sequence of points) starting at node $(x,y,z)$ and including all nodes $(x',y,z)$ with $x'\geq x$. The {\em weight} of this segment is $\sum_{x'\geq x}f(x',y,z)$. We define segments across the other dimensions analogously. The following problem\footnote{The reader may notice that the {\sc 3Segments} problem is a maximum weight independent set problem in disguise.} is essentially equivalent to the auction design problem:

\begin{defn}[{\sc 3Segments}]\label{defn:segments}
Given a joint discrete probability distribution $\phi$ supported on $\mathcal{G}$, which induces a set of segments on $\mathcal{G}$ as described above, find a subset of non-intersecting segments with maximum sum of weights.
\end{defn}

\begin{lem}\label{lem:MWIS23S}
The problem {\sc 3Segments($\phi$)} is equivalent to {\sc 3OptimalAuctionDesign$(\phi)$} via approximation-preserving reductions.
\end{lem}

\begin{proof}  {\em (Sketch:)}  The correspondence between solutions of the two problems is rather immediate, with the $x$-segment (respectively, $y$-segment, $z$-segment) at point $(x,y,z)$ included in the output set of segments if and only if $A[x,y,z] = 1$ (respectively, 2, 3) and $A[x',y,z]=0$ for all $x'<x$ (respectively, $A[x,y',z]=0$ for all $y'<y$ and $A[x,y,z']=0$ for all $z'<z$).
\end{proof}

\subsection{The construction}
\label{app:construction}
It suffices to show that {\sc 3Segments} is NP-hard to approximate for some constant; we do that by reducing from {\sc 3CatSat}, a  special case of {\sc 3Sat}:

\begin{defn}\label{defn:3catsat}
Let {\sc 3CatSat} be the {\sc 3Sat} problem with the input formula restricted to be of the following form. We have three types (categories) of variables $\{x_i\}_{i=1\ldots n_x},\{y_i\}_{i=1\ldots n_y}$ and $\{z_i\}_{i=1\ldots n_z}$, i.e. a total of $n=n_x+n_y+n_z$ variables, and $m$ clauses of the following form: every clause has at most one literal from every type (e.g. $(\bar{x_2}\vee\bar{y_3}\vee z_1)$ or $(\bar{x_5}\vee z_7)$).
\end{defn}

\begin{lem}\label{lem:3catsat}
{\sc 3CatSat} is NP-hard to approximate better than 79/80.
\end{lem}
\begin{proof}
We reduce from {\sc Max3Sat}. In order to turn an instance of {\sc Max3Sat} to an instance of {\sc 3CatSat} it suffices to create three copies for every occurrence of a variable: for variable $x$ we create $x_1,x_2,x_3$ and include the clauses $(\overline{x}_1\lor x_2),(\overline{x}_2\lor x_3),(\overline{x}_3\lor x_1)$. This increases the number of clauses by at most $3\bar{n}$, where $\bar{n}$ is the number of literals in the formula. Given that {\sc Max3Sat} is hard to approximate better than 7/8, it follows immediately that {\sc 3CatSat} is hard to approximate better than $\frac{7/8m+3\bar{n}}{m+3\bar{n}}$. Noticing that $\bar{n}\leq3m$ we get the desired approximation factor as an upper bound on this expression by picking $\bar{n}=3m$.
\end{proof}

We start with some intuition about the reduction. The instance of {\sc 3Segments} we create has three types of segments: literal segments, clause segments and scaffolding segments.

\begin{itemize}
\item  {\em Literal segments} are used to model truth assignments on variables; they ensure that every variable is assigned exactly one of the two possible truth values and that this assignment is consistent across all appearances of literals of this particular variable.
\item  {\em Clause segments} model the truth assignment to literals of a particular clause; we create one such clause segment for every literal that appears in a clause and we make them intersect with the literal segments of those literals; the idea is that if the clause is satisfied we will be able to pick at least one clause segment per clause because the corresponding literal segment will not be picked. Moreover, we cannot pick two or more clause segments per clause, since they will all intersect with each other.
\item {\em Scaffolding segments}\footnote{Rather, scaffolding points, as these segments have zero length.} ensure that, for some points, there are literal or clause segments that extend only to one of the three possible directions: In particular, given some point $(x,y,z)$ with positive marginal profit contribution for more than one players (for example when $f(x,y,z)>0$ and $g(x,y,z)>0$) we may want to have only $x$ or only $y$-segments starting from $(x,y,z)$; scaffolding points make sure this is the case.
\end{itemize}

We show how to construct a probability distribution $\phi$ that serves as the input of {\sc 3Segments}, given an instance of {\sc 3CatSat}. In what follows we use $\hat{n}$ to denote $\max\{n_x,n_y,n_z\}$.  The support\footnote{The support is actually a subset of this; these are all the points (values of the players) with {\em potentially} non-zero probabilities. This will become clear in the actual construction.} of $\phi$ is:  $\left\{h(i)| i=1,\ldots,\hat{n}+2m+4\right\}^3$, for an appropriate choice of the values $h(i)$ which we will fix later; for now all we assume is that $h$ is an increasing function of $i$. The size of the support is at most $(\hat{n}+2m+4)^3$, so this is clearly a polynomial time construction.

We shall abuse notation and write $\phi(x,y,z)$ instead of $\phi\left(h(x),h(y),h(z)\right)$ when there is no ambiguity. We shall also refer to the sub-matrices $\phi(i,\cdot,\cdot),\phi(\cdot,j,\cdot),\phi(\cdot,\cdot,k)$ as the ``planes'' $x=i,y=j$ and $z=k$ respectively.

 The construction goes as follows: we start with an all-zero matrix $\phi$ of the above size. Consider an arbitrary ordering of the clauses $1$ through $m$. Suppose the $l$-th clause is of the form $(x_i\vee y_j \vee z_k)$ where $x_i,y_j$ and $z_k$ can be either positive or negative literals, and $i\,\,(\text{resp. }j,k)=1,\ldots,n_x\,\,(\text{resp. }n_y,n_z)$ (the same construction also works for clauses with less than 3 variables). For this clause we introduce the following literal segments (1,2,3) and clause segments (4). In what follows we will first set the probability mass of the point that is the apex of each segment, and will later show how to use scaffolding-points to ensure that there is only one segment starting at each such point, towards the appropriate direction.

\begin{enumerate}
\item If $pos(x_i)$  then $\phi\left(i+1,\hat{n}+2,\hat{n}+2+l\right) = c_1$ and $\phi\left(i+1,\hat{n}+2+m+l,\hat{n}+2\right) = c_1$ 

(These points are intended to be the apices of a $y$-segment and a $z$-segment respectively.)

else $\phi\left(i+1,\hat{n}+2+l,\hat{n}+2\right) = c_1$ and $\phi\left(i+1,\hat{n}+2,\hat{n}+2+m+l\right) = c_1$

(These points are intended to be the apices of a $z$-segment and a $y$-segment respectively.)

\item
If $pos(y_j)$  then  $\phi\left(\hat{n}+2+l,j+1,\hat{n}+2\right) = c_1$ and $\phi\left(\hat{n}+2,j+1,\hat{n}+2+m+l\right) = c_1$, 

(These points are intended to be the apices of a $z$-segment and an $x$-segment respectively.)

else $\phi\left(\hat{n}+2,j+1,\hat{n}+2+l\right) = c_1$ and $\phi\left(\hat{n}+2+m+l,j+1,\hat{n}+2\right) = c_1$

(These points are intended to be the apices of an $x$-segment and a $z$-segment respectively.)
\item If $pos(z_k)$  then $\phi\left(\hat{n}+2,\hat{n}+2+l,k+1\right) = c_1$ and $\phi\left(\hat{n}+2+m+l,\hat{n}+2,k+1\right) = c_1$, 

(These points are intended to be the apices of an $x$-segment and a $y$-segment respectively.)
else $\phi\left(\hat{n}+2+l,\hat{n}+2,k+1\right) = c_1$ and $\phi\left(\hat{n}+2,\hat{n}+2+m+l,k+1\right) = c_1$

(These points are intended to be the apices of a $y$-segment and an $x$-segment respectively.)
\item $\phi\left(1,\hat{n}+2+l,\hat{n}+2+l\right) = \phi\left(\hat{n}+2+l,1,\hat{n}+2+l\right)  = \phi\left(\hat{n}+2+l,\hat{n}+2+l,1\right) = c_2$

(These points are intended to be the apices of an $x$-segment, a $y$-segment and a $z$-segment respectively.)
\end{enumerate}

Every positive occurrence of a variable of type, say, $x$ results in the following two segments: a {\em positive} literal segment starting at $\left(i+1,\hat{n}+2,\hat{n}+2+l\right)$ that intersects with the corresponding clause segment starting at $\left(1,\hat{n}+2+l,\hat{n}+2+l\right)$, and a literal segment starting at $\left(i+1,\hat{n}+2+m+l,\hat{n}+2\right)$ that does not intersect with any clause segment; this is also called a {\em dummy-negative} literal segment. Negative occurrences of variables analogously result in {\em negative}  and {\em dummy-positive} literal segments. Dummy literal segments are introduced because --for reasons that will become apparent later-- we want to ensure that we have an equal number of positive and negative literal segments (when dummies are included).

The reduction relies on the fact that the {\bf only} intersections involving literal and clause segments will be between literal segments of literals that are negations of each other, between clause segments of the same clause and between clause segments and their corresponding literal segments. To ensure that we need the aforementioned scaffolding segments; these will ensure that the only segments starting from the points defined above as having probability masses $c_1$ and $c_2$ (henceforth called $c_1$ and $c_2$ points), are the desired literal and clause segments (in other words we want those $c_1$ and $c_2$ points to be the apices of {\em only} the segments mentioned above):

\begin{enumerate}
\item We first ensure that there is exactly one segment starting from every point $c_2$, which is perpendicular to the plane  $x,$(or $y, z$ ) $=1$; in other words we ensure that there are no other segments starting at $c_2$ that lie on the plane, by introducing:  $\phi\left(1,\hat{n}+l+2,\hat{n}+2m+3\right)= \phi\left(1,\hat{n}+2m+3,\hat{n}+l+2\right)=\phi\left(\hat{n}+l+2,1,\hat{n}+2m+3\right)=\phi\left(\hat{n}+2m+3,1,\hat{n}+l+2\right)=\phi\left(\hat{n}+l+2,\hat{n}+2m+3,1\right)=\phi\left(\hat{n}+2m+3,\hat{n}+l+2,1\right)=c_3,$ 
for all $l=1,\ldots,m$, with the requirement that:

$$h(\hat{n}+l+2)\cdot(c_2+c_3)<h(\hat{n}+2m+3)\cdot c_3,\,\,\forall l=1,\ldots,m$$

The above requirement forces the corresponding marginal profit contribution function to be negative at any $c_2$ point, for some appropriately chosen player-direction: For example, we want point $\left(1,\hat{n}+2+l,\hat{n}+2+l\right)$ to be the apex of an $x$-segment only, and no $y$ or $z$-segment should start at this point. We achieve this by including the points $\phi\left(1,\hat{n}+l+2,\hat{n}+2m+3\right)= \phi\left(1,\hat{n}+2m+3,\hat{n}+l+2\right)=c_3$; the above inequality then ensures that $$g\left(1,\hat{n}+2+l,\hat{n}+2+l\right)=h\left(1,\hat{n}+2+l,\hat{n}+2+l\right)=0$$ and therefore there are no $y$ or $z$-segments starting at this point.

\item We next ensure that there are no segments starting from $c_1$ that go along row or column $\hat{n}+2$, by introducing:  $\phi\left(i+1,\hat{n}+2,\hat{n}+2m+3\right) =\phi\left(i+1,\hat{n}+2m+3,\hat{n}+2\right) = \phi\left(\hat{n}+2,j+1,\hat{n}+2m+3\right) =\phi\left(\hat{n}+2m+3,j+1,\hat{n}+2\right) = \phi\left(\hat{n}+2,\hat{n}+2m+3,k+1\right) =\phi\left(\hat{n}+2m+3,\hat{n}+2,k+1\right) = c_4,$ for $i=1,\ldots,n_x, j=1,\ldots,n_y, k=1,\ldots,n_z$, with the requirement that:

$$h(\hat{n}+l+2)\cdot(2mc_1+c_4)<h(\hat{n}+2m+3)\cdot c_4,\,\,\forall l=1,\ldots,2m$$

because we can have at most $m$ occurrences of any variable, and hence at most $2m$ $c_1$-entries on any given level. The rationale behind the inequality above is the same as in case (1) above.

\item Finally, we ensure that there are no segments starting from a point $c_1$ on plane $x=i$ (resp. $y=i$, $z=i$), for $i=\hat{n}+3,\ldots,\hat{n}+2m+2$, perpendicular to the plane $x$ (resp. $y,z$), by introducing: $\phi\left(\hat{n}+2m+4,i,\hat{n}+2\right)=\phi\left(\hat{n}+2m+4,\hat{n}+2,i\right)=\phi\left(\hat{n}+2,\hat{n}+2m+4,i\right)=\phi\left(i,\hat{n}+2m+4,\hat{n}+2\right)=\phi\left(i,\hat{n}+2,\hat{n}+2m+4\right)=\phi\left(\hat{n}+2,i,\hat{n}+2m+4\right)=c_5,$ for all $i=\hat{n}+3,\ldots,\hat{n}+2m+2$, with the requirement (along the same lines as above) that:

$$ \hat{n}\cdot h(\hat{n}+1)\cdot c_1< h(\hat{n}+2m+4)\cdot c_5$$

\end{enumerate}

The following values for $h(\cdot)$ and the constants\footnote{In order to have a proper probability distribution these constants need to be normalized by their sum.} satisfy all of the constraints above:

\[ h(i) = \left\{ \begin{array}{ll}
         1 +\frac{i-1}{\hat{n}+2m+1} & \mbox{for $i = 1,\ldots,\hat{n}+2m+2$}\\
        4 & \mbox{for $i = \hat{n}+2m+3$}\\
        5 & \mbox{for $i = \hat{n}+2m+4$}\end{array} \right. \] 
        
and 
$$c_1=\frac{1}{\hat{n}^2m},\,\,\,\, c_2=1,\,\,\,\,c_3=1,\,\,\,\,c_4=\frac{3}{\hat{n}^2},\,\,\,\,c_5=\frac{2}{5\hat{n}m}$$
 
Our main result for this section is then the following:

\begin{lem}\label{thm:main_reduction}
It is NP-hard to approximate {\sc 3Segments} better than 0.05\%.
\end{lem}
\begin{proof}
We first note that regardless of the instance of {\sc3CatSat} we are reducing from, we can always obtain a fixed profit for {\sc 3Segments} from the scaffolding points; by picking the most profitable segments starting at each of the scaffolding points we get a total profit of:
$$F = 6m\cdot h(\hat{n}+2m+3) c_3 + 2n\cdot h(\hat{n}+2m+3) c_4+12m\cdot h(\hat{n}+2m+4) c_5$$

Let $\bar{n}$ be the total number of literal occurrences in the formula. We then have the following:

\begin{itemize}
\item If the {\sc 3CatSat} formula is satisfiable then the profit of {\sc 3Segments} is exactly: 

\begin{equation}\label{eq:cost1}
m\cdot h(1) c_2 + \bar{n}\cdot h(\hat{n}+2) c_1 + F
\end{equation}

To see this first consider the following way to pick the literal segments according to the truth values assigned to the corresponding variables: if a variable is set to {\bf true}\footnote{The other case is completely symmetrical.} we include its {\bf negative} literal segments (and the corresponding dummy-positive literal segments). Notice that --thanks to the dummy literal segments-- there is an equal number of positive and negative literal segments (dummies included), with totally $2\bar{n}$ of them; we include exactly half of them for every variable (either the positive or the negative ones), so the total profit from these segments is exactly $\bar{n}\cdot h(\hat{n}+2) c_1$. 

Moreover, since the formula is satisfiable, at least one literal per clause is satisfied; if this is a positive\footnotemark[\value{footnote}] literal, then the variable has been set to true so the literal segments for this variable included in our {\sc 3Segments} solution will be the negative ones. However since the variable appears as a positive literal at this clause, our construction ensures that the corresponding clause segment intersects only with the positive literal segment of this variable. Therefore, since the clause segment does not intersect with the negative literal segment (that we have already included in our {\sc 3Segments} solution), we can include the clause segment in our solution as well. Each one of these clause segments contributes $h(1) c_2$; noticing that we cannot include more than one clause segment from each clause (because they intersect) we get that their total contribution is exactly $m\cdot h(1) c_2$.

\item If the optimal assignment for {\sc 3CatSat} satisfies at most $\rho$ of the clauses, then the optimal profit for {\sc 3Segments} is at most:

\begin{equation}\label{eq:cost2}
\rho m\cdot h(1) c_2 + \bar{n}\cdot h(\hat{n}+2) c_1 + F
\end{equation}

To prove this we show how to transform a solution of profit $\geq \rho m\cdot h(1) c_2 + \bar{n}\cdot h(\hat{n}+2) c_1 + F$ to a truth assignment that satisfies more than $\rho$ of the clauses. First notice that in order to achieve such a profit we must include exactly $\bar{n}$ literal segments and exactly one clause segment per clause for $\rho$ of the clauses: we cannot include more without having intersecting segments and we cannot include any less and achieve the same profit. These $\bar{n}$ literal segments correspond to a truth value assignment to the variables of the formula as described above: we set every variable whose positive (resp. negative) literal segments are included to false (resp. true). The claim follows by noticing that this truth value assignment satisfies every clause for which a clause segment was included in the {\sc 3Segments} solution, i.e. for a fraction $\rho$ of the clauses.
\end{itemize}
From the way we set the constants and $h(\cdot)$, it follows that expression (\ref{eq:cost1}) is $\geq 25m$ and expression (\ref{eq:cost2}) is $\leq (24+\rho)m + \delta$ for some constant $\delta \ll 0.1$, for sufficiently large $m$ and $n$. Using the fact from Lemma \ref{lem:3catsat} that $\rho=79/80$ we get an approximation ratio of 1999/2000.
\end{proof}

\begin{figure}[H]
\centering
\includegraphics[scale = 0.4]{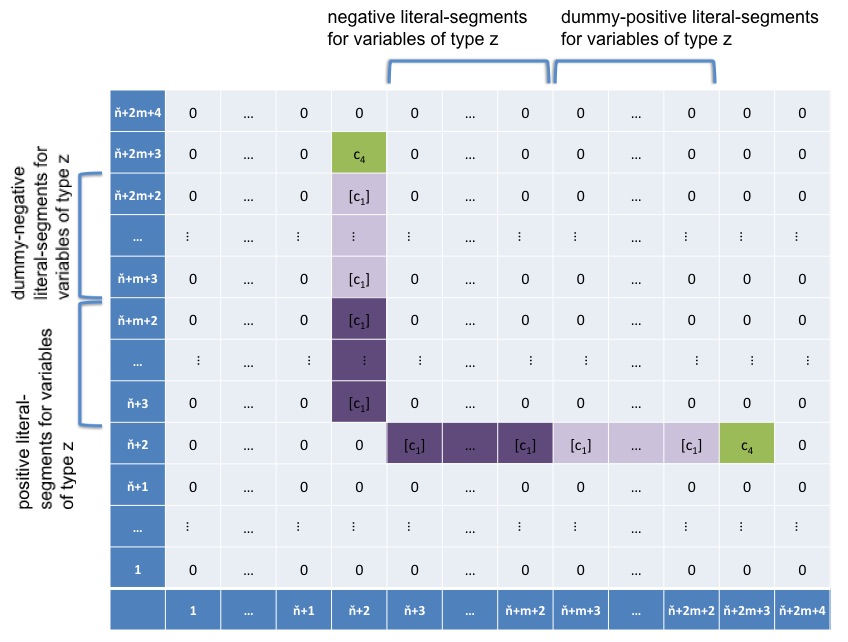}
    \caption{\small The $xy$-plane for $z = k+1,\,k=1,\ldots,n_z$, contains the literal segments of variable $z_k$. By $[\cdot]$ we mean that this point does not appear at all levels $z$. Every such level must have an equal number of $c_1$-entries in line $\hat{n}+1$ and column $\hat{n}+1$; also notice that it is not possible to simultaneously have non-zero entries both for row and column $i$ at the same level $z$. }
\label{figure4}
\end{figure}

\begin{figure}[H]
\centering
\includegraphics[scale = 0.4]{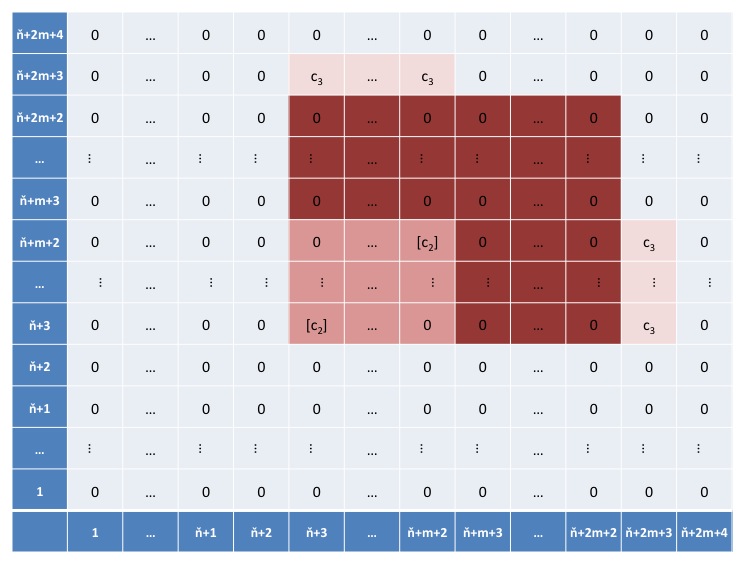}
  \caption{\small The $xy$-plane for $z=1$. The exact position of every $c_2$ depends on the clause in which the variable appears: every time a variable $z_k$ appears in clause $l$, we introduce two literal segments lying on level $z=k+1$, through $c_1$-points, and a clause segment perpendicular to this plane, through a $c_2$-point, so that it intersects with the non-dummy literal segment.}
\label{figure5}
\end{figure}  

\begin{figure}[H]
\centering
\includegraphics[scale = 0.4]{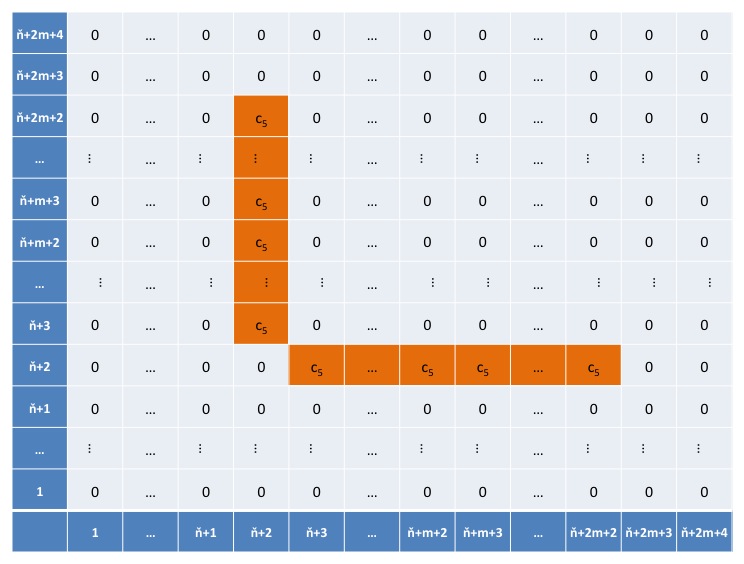}
  \caption{\small The $xy$-plane for $z=\hat{n}+2m+4$.}
  \label{figure6}
\end{figure}  

\begin{figure}[H]
\centering
\includegraphics[scale = 0.5]{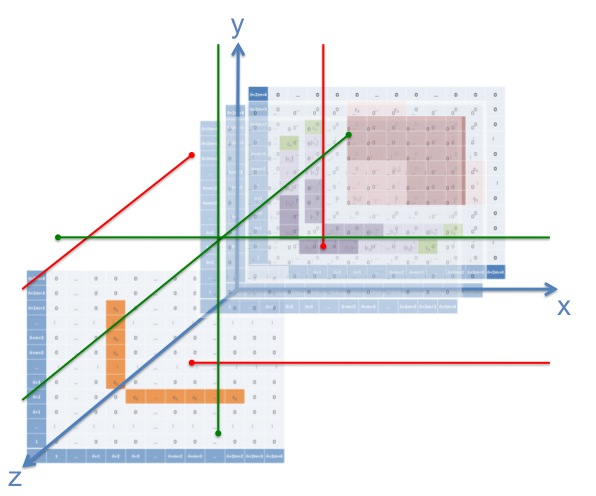}
  \caption{\small The construction in three dimensions. The red lines depict literal segments and the green lines clause segments as they are positioned in the space.}
  \label{figure6}
\end{figure}  

Combining Lemmas~\ref{lem:MWIS23S} and~\ref{thm:main_reduction} immediately yields our main theorem for this section:

\begin{thm}
It is NP-hard to approximate {\sc 3OptimalAuctionDesign} better than 0.05\%.
\end{thm}

\section{Discussion and Open Problems}
Even though in this paper we focused on deterministic mechanisms, our geometric characterization has interesting consequences for randomized mechanisms. Remember that for the discrete case the optimal (deterministic) mechanism immediately follows from solving the integer program of Problem C in Section~\ref{proof_main}, i.e. computing a maximum weight independent set in the corresponding $n$-partite graph. Our first observation is that the linear programming relaxation of this integer program corresponds to computing the optimal randomized mechanism. For two players, where the graph is bipartite and the integer program is totally unimodular, the optimum integer solution is also the optimum of the relaxed linear program. Therefore, for two bidders, the program of Problem C computes a deterministic mechanism that is optimal among all randomized mechanisms: this is reminiscent of Myerson's original result, where the deterministic mechanism obtained is optimal for the (larger) class of {\em Bayesian truthful} randomized mechanisms~\cite{RePEc:nwu:cmsems:362}. For a constant number of three or more bidders, the generalization of our geometric characterization yields polynomial-time algorithms for computing the optimal randomized auction, in sharp contrast with the intractability of computing the optimal deterministic auction, even for three bidders; of course for a large number of bidders, the size of this linear program may become exponentially large and therefore this approach is infeasible. For an alternative linear program that computes the optimal randomized auction for any number of bidders, when the distribution is given explicitly, the reader is referred to~\cite{DBLP:journals/corr/abs-1011-2413}.

An important open problem of this work is to close the gap between the best approximation algorithm known for the optimal auction problem (currently .60)~\cite{DBLP:journals/corr/abs-1011-2413} and the inapproximability bound (currently about $.9995$).  We believe that progress there is attainable.  Interestingly, our work implies an approximation of $2/n$ for $n$ players: before having the bidders announce their bids, the auctioneer looks at their joint distribution and privately runs the optimal auction for all possible pairs of players. Since solving for the optimal auction is nothing but a maximum weight independent set problem on the corresponding graph, it is easy to prove that the profit of the best of those ${n\choose2}$ auctions is at least $2/n$ of the overall profit. The auctioneer then rejects a priori all but the bidders who were part of the most profitable two-bidder auction and then runs it. The overall auction is obviously truthful as long as bidders are rejected before even submitting their bids. For 3 bidders this gives an approximation ratio of 2/3, improving over Ronen's auction~\cite{DBLP:conf/sigecom/Ronen01}, but for $n\geq4$ the approximation ratio drops below 1/2.  
\bibliographystyle{plain}
\bibliography{refs}

\end{document}